\documentclass[12pt]{article}

\usepackage[usenames]{color}
\usepackage{amssymb}
\usepackage{graphicx}
\usepackage{amscd}

\usepackage[colorlinks=true,
linkcolor=webgreen,
filecolor=webbrown,
citecolor=webgreen]{hyperref}

\definecolor{webgreen}{rgb}{0,.5,0}
\definecolor{webbrown}{rgb}{.6,0,0}

\usepackage{color}
\usepackage{fullpage}
\usepackage{float}

\usepackage{amsmath}
\usepackage{amsthm}
\usepackage{amsfonts}

\newtheorem{theorem}{Theorem}
\newtheorem{lemma}{Lemma}
\newtheorem{corollary}{Corollary}

\theoremstyle{definition}\newtheorem{remark}{Remark}
\theoremstyle{definition}\newtheorem{problem}{Problem}

\begin{document}
\title{Finite test sets for morphisms which are square-free on some of Thue's square-free ternary words}
\author{James D. Currie\\
Department of Mathematics and Statistics\\ 
University of Winnipeg\\
Winnipeg, Manitoba R3B 2E9, Canada}
\maketitle

\begin{abstract}
Let $S$ be one of $\{aba,bcb\}$ and $\{aba, aca\}$, and let $w$ be an infinite square-free word over $\Sigma=\{a,b,c\}$ with no factor in $S$. 
Suppose that $f:\Sigma\rightarrow T^*$ is a non-erasing morphism. Word $f(w)$ is square-free if and only if  $f$ is square-free on factors of $w$ of length 7 or less. 
\end{abstract}

The papers of Axel Thue on square-free words \cite{thue,thue2} are foundational to the area of combinatorics on words. A word $w$ is {\em square-free} if we cannot write $w=xyyz$, where $y$ is a non-empty word. The longest square-free words over the 2-letter alphabet $\{a,b\}$ are $aba$ and $bab$, each of length 3, but Thue showed that arbitrarily long square-free words exist over a 3-letter alphabet. Infinite square-free words over finite alphabets are routinely encountered in combinatorics on words, and are frequently used as building blocks in  constructions. (See \cite{loth97,loth02} for further background and definitions.)

Let $w$ be an infinite square-free word over $\Sigma=\{a,b,c\}$. Thue showed that $w$ must contain every length 2 square-free word over $\Sigma$ as a factor. However, Thue showed that the same is not true for length 3 square-free words over $\Sigma$. For each of $S_1=\{aba,bcb\}, S_2=\{aba, aca\}$, and $S_3=\{aba,bab\}$, he constructed an infinite square-free word over $\Sigma$ with no factor in $S_i$. The infinite square-free word over $\Sigma$ with no factor in $S_1$ has been called {\em vtm} (for `variation of Thue-Morse'), and has been used in constructions in several papers \cite{blanchetsadri,currie91,currieetal}. Various constructions involve the necessity of showing $f$({\em vtm}) to be square-free for particular morphisms $f$. In this paper, we give a simple testable characterization of morphisms $f$ such that $f(${\em vtm}) is square-free; we do the same in the case where {\em vtm} is replaced by an infinite square-free word over $\Sigma$ with no factors in $S_2$.

From now on, fix $S$ to be one of $S_1$ and $S_2$, and let $w$ be an infinite square-free word over $\Sigma$ with no factor in $S$. We establish the following:

\begin{theorem}\label{main}
Suppose that $f:\Sigma\rightarrow T^*$ is a non-erasing morphism. Word $f(w)$ is square-free if and only if  $f$ is square-free on factors of $w$ of length 7 or less. 
\end{theorem}

Our theorem says that to establish square-freeness of $f(w)$, one need only check $f$ for square-freeness on a finite test set. A variety of similar theorems were proved by Crochemore \cite{crochemore}; in particular, a morphism $g$ defined on $\Sigma$ preserves square-freeness exactly when it preserves square-freeness on words of $\Sigma^*$ of length at most 5. Note that the theorem of the present paper tests a weaker property; while $aba$ is square-free, we do not require 
$f(aba)$ to be square-free, for example. Finite test sets for morphisms preserving overlap-freeness have also been well-studied \cite{richomme}.

The theorem suggests a couple of open problems:

\begin{problem} Is the constant 7 in Theorem~\ref{main} optimal?
\end{problem}
\begin{problem}\label{S_3}
Suppose $W$ 
is an infinite square-free word over $\Sigma$ with no factor in $S_3$. Is there a constant $k$ such that for any non-erasing morphism $f$ on $\Sigma$, $f(W)$ is square-free if and only if it is square-free on factors of $W$ of length at most $k$?
\end{problem}

Write
$$w=a_0a_1a_2a_3\cdots, a_i\in\Sigma.$$

Suppose that $f:\Sigma\rightarrow T^*$ is a non-erasing morphism which is square-free on factors of $w$ of length 7 or less. 
\begin{lemma}\label{at most 3}Suppose $f(\chi)$ is a factor of $f(x)$, where $x\in\Sigma$ and $\chi$ is a factor of $w$. Then $|\chi|\le 3$.
\end{lemma}
\begin{proof} If $x$ is a letter of $\chi$,  and $f(\chi)$ is a factor of $f(x)$, $$|f(x)|\le|f(\chi)|\le|f(x)|.$$
Since $f$ is non-erasing, this  forces $x=\chi$, giving $|\chi|=1$. 

If $x$ is not a letter of $\chi$,  then $\chi$ is a square-free word over a two-letter alphabet, so that $|\chi|\le 3.$ 
\end{proof}

\begin{lemma}\label{no prefix}Suppose that
$f(x)$ is a prefix or suffix of $f(y)$, where $x,y\in\Sigma$. Then $x=y$.
\end{lemma}
\begin{proof}
We give the proof where $f(x)$ is a prefix of $f(y)$. (The other case is similar.) Suppose
$x\ne y$. Then
$xy$ must be a factor of $w$, and $xy$ is square-free. However, $f(xy)$ begins with the square $f(x)f(x)$, contradicting the square-freeness of $f$ on factors of $w$ of length at most 7.\end{proof}

\begin{lemma}\label{suffix} There is no solution to the equation 
$u=sf(\chi xyz)p$ such that $t,\alpha, x,y,z,\beta\in\Sigma$, $u$ is a suffix of $f(t)$, $s$ is a suffix of $f(\alpha)$, $p$ is a non-empty prefix of $f(\beta)$, $\alpha\chi xyz\beta$ is a factor of $w$.
\end{lemma}
\begin{proof}
Suppose $t,\alpha,\chi, x,y,z,\beta, u,s,p$ are such that  $t,\alpha,x,y,z,\beta\in\Sigma$, $u$ is a suffix of $f(t)$, $s$ is a suffix of $f(\alpha)$, $p$ is a non-empty prefix of $f(\beta)$, $\alpha\chi xyz\beta$ is a factor of $w$, and $u=sf(\chi xyz)p$.

Since $p$ is a non-empty prefix of $f(\beta)$, but also a suffix of $f(t)$, we see that $f(t\beta)$ contains square $pp$, so that $t\beta$ is not a factor of $w$. This forces $t=\beta$. On the other hand, since $z\beta$ is a factor of $w$, we conclude that $z\ne t$. Again, $f(z)p$ is a prefix of $f(z\beta)$, but also a suffix of $f(t)$, so that $f(tz\beta)$ contains a square. Since $yz\beta$ is a factor of $w$, it follows that $y\ne t$. Similarly, $f(tyz\beta)$ contains a square, but $xyz\beta$ is a factor of $w$, so that that $x\ne t$. Finally, we see that
$f(txyz\beta)$ contains a square.
Let $\delta$ be the last letter of $\alpha\chi$. We conclude that $\delta\ne t$. However, now $\delta xyz$ is a squarefree word of length 4 over the two-letter alphabet $\Sigma-\{t\}$. This is impossible.
\end{proof}

The symmetrical lemma is proved analogously:

\begin{lemma}\label{prefix} There is no solution to the equation 
$u=sf(xyz\chi)p$ such that $t,\alpha, x,y,z,\beta\in\Sigma$, $\chi\in\Sigma^*$, $u$ is a prefix of $f(t)$, $s$ is a non-empty  suffix of $f(\alpha)$, $p$ is a prefix of $f(\beta)$, $\alpha xyz\chi\beta$ is a factor of $w$.
\end{lemma}
\begin{theorem}\label{abc}Suppose that $f(w)$ contains a non-empty square. Then $w$ contains a factor  $\alpha z \beta z \gamma$, $\alpha,\beta,\gamma\in\Sigma$, $\alpha,\gamma\ne \beta$, $|z|\ge 3$, such that  $\alpha\beta\gamma$ is not a factor of $w$.
\end{theorem}

The proof of Theorem~\ref{abc} will use several lemmas. Suppose that $f(w)$ contains a non-empty square $xx$, with $|x|$ as short as possible. Write
$f(w)=uxxv$, such that
$$u=A_0A_1\cdots A_i'$$
$$ux=A_0A_1\cdots A_j'$$
$$uxx=A_0A_1\cdots A_k'$$
where $i\le j\le k$ are non-negative integers, and for each non-negative integer $\ell$, $A_\ell=f(a_\ell)$, and $A_\ell'$ is a prefix of $A_\ell$, but $A_\ell'\ne A_\ell$. This notation is not intended to exclude the possibilities that $i=0$, $i=j$ and/or $j=k$.

\begin{remark}
Since $f$ is square-free on factors of $w$ of length at most 7, but $f(a_i\cdots a_k)$ contains the square $xx$, we must have $k-i\ge 7$.
\end{remark}
\begin{remark}
We cannot have $i=j$; otherwise suffix $x$ of $ux$ is a factor of $A_j=A_i$, and $A_{i+1}\cdots A_{k-1}$ is a factor of suffix $x$ of $uxx$. Then, $f(a_{i+1}a_{i+2}\cdots a_{k-1})$ is a factor of $f(a_i)$, forcing $((k-1)-(i+1)+1\le 3$ by Lemma~\ref{at most 3}, so that $k-i\le 4$, a contradiction. Reasoning, in the same way, we show that $i<j<k$.

For $\ell\in\{i,j,k\}$, let $A_\ell^{\prime\prime}$ be the suffix of $A_\ell$ such that $A_\ell=A_\ell'A_\ell^{\prime\prime}$. By our choice of $A_\ell'$, $A_\ell^{\prime\prime}\ne\epsilon.$
Then 
\begin{equation}\label{x=x}x=A_i^{\prime\prime}A_{i+1}\cdots A_{j-1}A_j'=A_j^{\prime\prime}A_{j+1}\cdots A_{k-1}A_k'.\end{equation}
\end{remark}

\begin{remark}
Since $k-i\ge 7$, we must have $k-j-1\ge 3$ and/or $j-i-1\ge 3.$
\end{remark}
\begin{lemma}We must have $|A_i^{\prime\prime}|+|A_k'|\le|x|$ and
$|A_j^{\prime\prime}|+|A_j'|\le|x|$.
\end{lemma}
\begin{proof} We give the proof that $|A_j^{\prime\prime}|+|A_j'|\le|x|$. (The proof
of the other assertion is similar.) Suppose for the sake of getting a contradiction that $|A_j^{\prime\prime}|+|A_j'|>|x|$. Then
$|A_j^{\prime\prime}|>|x|-|A_j'|=|A_i^{\prime\prime}A_{i+1}\cdots A_{j-1}|$. It follows that $A_j^{\prime\prime}=A_i^{\prime\prime}A_{i+1}\cdots A_{j-1}A_j^{\prime\prime\prime}$ for some non-empty prefix $A_j^{\prime\prime\prime}$ of $A_j$. Similarly, one shows that $A_j^{\prime}=A_j^{\prime\prime\prime\prime}A_{j+1}\cdots A_{k-1}A_k'$ for some non-empty suffix $A_j^{\prime\prime\prime\prime}$ of $A_j$. Now $|a_{i+1}\cdots a_{j-1}|=(j-1)-(i+1)+1=j-i-1$, and  $|a_{j+1}\cdots a_{k-1}|=(k-1)-(j+1)+1=k-j-1$. However, either 
$j-i-1\ge 3$, or $k-j-1\ge 3$. If $j-i-1\ge 3$, then
$A_j^{\prime\prime}=A_i^{\prime\prime}A_{i+1}\cdots A_{j-1}A_j^{\prime\prime\prime}$ for some non-empty prefix $A_j^{\prime\prime\prime}$ of $A_j$ which contradicts 
 Lemma~\ref{suffix}, letting $s=A_i^{\prime\prime}$, $\alpha = a_i$,
$\chi= a_{i+1}\cdots a_{j-4}$, $xyz=a_{j-3}a_{j-2}a_{j-1}$, $p=A_j^{\prime\prime\prime}$, $\beta=a_j$.

In the case where $k-j-1\ge 3$, we get the analogous contradiction using
 Lemma~\ref{prefix}.
\end{proof}
\begin{lemma}\label{lineup} We have $A_j^{\prime}=A_k^{\prime}$,  $A_i^{\prime\prime}=A_j^{\prime\prime}$,
 $j-i=k-j$, and $A_{i+\ell}=A_{j+\ell}$, $1\le \ell\le j-i-1$.
\end{lemma}
\begin{proof}
To begin with we show that $A_j^{\prime}=A_k^{\prime}$. Since both words are suffixes of $x$, it suffices to show that 
$|A_j^{\prime}|=|A_k^{\prime}|$. 
Suppose for the sake of getting a contradiction that $|A_k^{\prime}|<|A_j^{\prime}|$. 

By the previous lemma, $|A_j^{\prime}|\le|A_{j+1}\cdots A_{k-2}A_{k-1}A_k'|$. Let $m$ be greatest such that $|A_j^{\prime}|\le|A_m\cdots A_{k-1}A_k'|.$ Thus $j+1\le m\le k-1$, and
$$|A_{m+1}\cdots A_{k-1}A_k'|<|A_j^{\prime}|\le |A_mA_{m+1}\cdots A_{k-1}A_k'|.$$
	It follows that  $A_j^{\prime}=A_m^{\prime\prime}A_{m+1}\cdots A_{k-1}A_k'$, for some non-empty suffix $A_m^{\prime\prime}$ of  
$A_m$. Therefore, $f(a_ma_j)$ contains the square factor $A_m^{\prime\prime}A_m^{\prime\prime}$, forcing $a_m=a_j$. If
$A_m^{\prime\prime}=A_m$, then $A_m$ is a prefix of $A_j$, and Lemma~\ref{no prefix} forces  $A_m=A_j^{\prime}=A_j$. This is contrary to our choice of $A_j^{\prime}$. We may therefore  write $A_m=A_m'A_m^{\prime\prime}$ where $A_m'$ is non-empty.

Now $A_m'$ is a suffix of $A_i^{\prime\prime}A_{i+1}\cdots A_{j-1}$. Let $n$ be greatest such that $A_m'$ is a suffix of $A_nA_{n+1}\cdots A_{j-1}$; thus $i\le n\le j-1$. 
It follows that a non-empty prefix $A_m^{\prime\prime\prime}$ of $A_m'$ is a suffix of $A_n$; this implies $f(a_na_m)$ contains a square, whence $a_n=a_m$.
Let $q\le j-1$ be greatest such that $a_q=a_m$; thus $i\le n\le q\le j-1$. Each of $A_m'$ and $A_q\cdots A_{j-1}$ is a suffix of $A_nA_{n+1}\cdots A_{j-1}$. Since $|A_m'|<|A_m|=|A_q|$, we conclude that $A_m'$ is a suffix of $A_q\cdots A_{j-1}$. By the choice of $q$, $a_{q+1}\cdots a_{j-1}$ does not have letter $a_m$ as a factor; it is thus a square-free word over a 2-letter alphabet, whence $|a_{q+1}\cdots a_{j-1}|\le 3$. However, $A_m'$ is a suffix of $A_q\cdots A_{j-1}$, so that the square
$A_m'A_m'$ is a factor of $A_q\cdots A_{j-1}A_m=A_q\cdots A_{j-1}A_j.$ This is impossible, since $a_q\cdots a_{j-1}a_j$ is a factor of $w$ of length at most 4.

The assumption $|A_k'|<|A_j'|$ leads to a contradiction. 
The assumption $|A_j'|<|A_k'|$ leads to a similar contradiction, and we conclude that $A_j'=A_k'$, as desired.

Next, we show that $j-i=k-j$ and $A_{i+\ell}=A_{j+\ell}$, $1\le \ell\le j-i-1$. Suppose that $j-i\le k-j$. (The other case is similar.) 
Suppose now that we have shown that for some $\ell$, $0\le \ell< j-i-1$  that
\begin{equation}\label{IHOP}A_{j-\ell}\cdots A_{j-1}A_j'=A_{k-\ell}\cdots A_{k-1}A_k'.\end{equation}
This is true when $\ell=0$; i.e., $A_j^{\prime}=A_k^{\prime}$.

From (\ref{x=x}), one of $A_{j-\ell-1}A_{j-\ell}\cdots A_{j-1}A_j'$ and $A_{k-\ell-1}A_{k-\ell}\cdots A_{k-1}A_k'$  
 is a suffix of the other. Together with (\ref{IHOP}), this implies that
 one of $A_{j-\ell-1}$ and
$A_{k-\ell-1}$ 
 is a suffix of the other.
By Lemma~\ref{no prefix}, this implies that $a_{j-\ell-1}=a_{k-\ell-1}$, and by combining this with (\ref{IHOP}),
\begin{equation}A_{j-\ell-1}\cdots A_{j-1}A_j'=A_{k-\ell-1}\cdots A_{k-1}A_k'.\end{equation}
By induction we conclude that
$$A_{j-\ell}\cdots A_{j-1}A_j'=A_{k-\ell}\cdots A_{k-1}A_k',0\le\ell\le j-i-1,$$

which implies $A_{j-\ell}=A_{k-\ell}$, $1\le \ell\le j-i-1$.
In particular, we note that
\begin{equation}\label{var x=x}A_{i+1}\cdots A_{j-1}=A_{k-j+i+1}\cdots A_{k-1}.\end{equation}

If we now have $k-j>j-i$, then $k-j+i>j$, and (\ref{x=x}) and (\ref{var x=x}) imply that
$A_i^{\prime\prime}=A_j^{\prime\prime}A_{j+1}\cdots A_{k-j+i}$. Then $A_{k-j+i}$ is a suffix of $A_i^{\prime\prime}$ and Lemma~\ref{no prefix} forces $A_{k-j+i}=A_i$. Then (\ref{x=x}) and (\ref{var x=x})
force $A_i=A_i^{\prime}$, contrary to our choice of $A_i^{\prime}$.
We conclude that $k-j=j-i$. From (\ref{x=x}) and (\ref{var x=x}) we conclude that $A_i^{\prime\prime}=A_j^{\prime\prime}$, as desired.
\end{proof}
\begin{proof}[Proof of Theorem \ref{abc}] By Lemma~\ref{lineup}, $w$ contains a factor
$\alpha z \beta z \gamma$, where $\alpha=a_i$, $\beta=a_j$, $\gamma=a_k$, $z=a_{i+1}\cdots a_{j-1}=a_{j+1}\cdots a_{k-1}$. This gives $|z|=j-i-1=k-j-1\ge 3$.

Since $w$ is square-free, we cannot have $a_i=a_j$; otherwise $w$ contains the square $(a_iz)^2$; similarly, $a_j\ne a_k$. To see that $\alpha\beta\gamma$ is not a factor of $w$, we note that $f$ is square-free on factors of $w$ of length at most 7, but $f(\alpha\beta\gamma)=A_iA_jA_k$ contains the square $(A_i^{\prime\prime} A_j')(A_j^{\prime\prime} A_k')$; this is a square since $A_i^{\prime\prime}=A_j^{\prime\prime}$ and $A_j'=A_k'$. Since $|A_iA_jA_k|=3\le 7$, we conclude that $a_ia_ja_k$ is not a factor of $w$.
\end{proof}
\begin{lemma}\label{aba,cbc} If $S=\{aba,cbc\}$, the only length 3 squarefree words over $\Sigma$  which are not factors of $w$ are $aba$ and $cbc$. In addition, $w$ contains no factor of the form $azbza$ or $czbzc$.
\end{lemma}
\begin{proof} Thue \cite{thue2} showed that a squarefree word over $\Sigma$ not containing $aba$ or $cbc$ as a factor contains every other length 3 square-free word as a factor. 

Suppose $w$ contains a factor $azbza$. Since $aba$ is not a factor of $w$, $z\ne \epsilon$. Since $az$ and $bz$ are factors of $w$, and hence square-free, the first letter of $z$ cannot be $a$ or $b$, and must be $c$. Similarly, the last letter of $z$ must be $c$. But $w$ contains $zbz$, and thus $cbc$. This is a contradiction. Therefore $w$ contains no factor $azbza$.

Replacing $a$ by $c$ and vice versa in the preceding argument shows that $w$ contains no factor $czbzc$.
\end{proof}

Combining Theorem 2 and Lemma~ref{aba,cbc}  gives this corollary:

\begin{corollary} If $S=\{aba,cbc\}$ and $f:\Sigma^*\rightarrow T^*$ is square-free on factors of $w$ of length at most 7, then $f(w)$ is squarefree.
\end{corollary}

\begin{lemma}\label{aba,aca} If $S=\{aba,aca\}$, the only length three squarefree words over $\Sigma$ which are not factors of $w$ are $aba$ and $aca$. In addition, $w$ contains no factor of the form $azbza$ or $azcza$, $|z|\ge 3$.
\end{lemma}
\begin{proof} Thue \cite{thue2} showed that a squarefree word over $\Sigma$ not containing $aba$ or $aca$ as a factor contains every other length 3 square-free word as a factor. 

Suppose $w$ contains a factor $azbza$, $|z|\ge 3$.  Since $az$ and $bz$ are factors of $w$, and hence square-free, the first letter of $z$ cannot be $a$ or $b$, and must be $c$. Similarly, the last letter of $z$ must  be $c$. However, since $az$ is a factor of $w$, but $aca$ is not, the second letter of $z$ cannot be $a$ and must be $b$. Write $z=cbz'c$. Then $azbza=acbz'cbcbz'ca$ contains the square $cbcb$, which is impossible. We conclude that $w$ contains no factor $azbza$, $|z|\ge 3$.

Replacing $c$ by $b$ and vice versa in the preceding argument shows that $w$ contains no factor $azcza$, $|z|\ge 3$.
\end{proof}

\begin{corollary} If $S=\{aba,aca\}$ and $f:\Sigma^*\rightarrow T^*$ is square-free on factors of $w$ of length at most 7, then $f(w)$ is square-free.
\end{corollary}

\noindent Combining Corollaries 1 and 2 establishes our main result.\vspace{.1in}

\noindent {\bf Theorem 1:}
{\it Suppose that $f:\Sigma\rightarrow T^*$ is a non-erasing morphism. Word $f(w)$ is square-free if and only if  $f$ is square-free on factors of $w$ of length 7 or less.}

\begin{remark} The square-free word $azbza$ where $z=cabcbac$ has no factors $aba$ or $bab$. It follows that any analogous theorem for $S_3$, with an analogous proof, would require us to replace 7 by a value of at least $|azbza|=17.$
\end{remark}


\begin{thebibliography}{3}
\bibitem{berstel}J. Berstel, {\em Axel Thue's Papers on Repetitions in Words: a Translation}. Publications du Laboratoire de Combinatoire et d'Informatique Math\'ematique {\bf 20} Universit\'e du Qu\'ebec \`a Montr\'eal (1995).
\bibitem{blanchetsadri}F. Blanchet-Sadri, J. Currie, N. Fox, N. Rampersad, Abelian complexity of fixed point of morphism $0\rightarrow 012, 1 \rightarrow 02, 2 \rightarrow 1$, {\em INTEGERS} {\bf 14} (2014), A11.
\bibitem{crochemore}M. Crochemore,
Sharp characterizations of squarefree morphisms
{\em Theoret. Comput. Sci.}, {\bf 18 (2)} (1982), 221--226.
\bibitem{currie91} J.  D.  Currie.   Which  graphs  allow  infinite  nonrepetitive  walks?, {\em Discrete  Math.},{ \bf 87}': 249--260, 1991.
\bibitem{currieetal}James Currie, Tero Harju, Pascal Ochem, Narad Rampersad, Some further results on squarefree arithmetic progressions in infinite words,  {\em arXiv}: 1901.06351.
\bibitem{loth97}M. Lothaire,  {\em Combinatorics on Words}. Cambridge University Press, Cambridge (1997).
\bibitem{loth02}M. Lothaire, {\em Algebraic Combinatorics on Words}. Cambridge University
\bibitem{richomme}
G. Richomme, F. Wlazinski,
Overlap-free morphisms and finite test-sets
{\em Discrete Appl. Math.}, {\bf 143} (2004), 92--109.
\bibitem{thue}  Axel Thue, \"{U}ber unendliche Zeichenreihen, 
{\em Norske Vid. Selsk. Skr. Mat. Nat. Kl.} {\bf 7} (1906), 1--22.
\bibitem{thue2} Axel Thue,  \"{U}ber die gegenseitige Lage gleicher Teile gewisser Zeichentreihen, {\em Norske Vid. Selske Skr. Mat. Nat. Kl.} {\bf 1} (1912), 1--67.

\end{thebibliography}
\end{document}